\documentclass[11pt]{article}

\usepackage{theorem}
\usepackage{amsmath,amssymb,amsfonts}
\usepackage{latexsym}
\usepackage{mathrsfs}
\usepackage{lscape}
\usepackage{enumerate}
\usepackage{hyperref}

  \newtheorem{theorem}{Theorem}[section]
  \newtheorem{corollary}[theorem]{Corollary}
   
  \newtheorem{lemma}[theorem]{Lemma}
  
  \newtheorem{definition}[theorem]{Definition}
  \newtheorem{example}[theorem]{Example}
  \newtheorem{remark}[theorem]{Remark}
  \numberwithin{equation}{section}

\newenvironment{proof}{\noindent  Proof:\ }{\hspace*{\fill} $\Box $\\}

\def\N{\mathbb{N}}
\def\Z{\mathbb{Z}}
\def\F{\mathbb{F}}
\def\P{\mathbb{P}}

\def\ord{\mathrm{ord}}

\def\codeC{\mathscr{C}}

\title{On quasi-cyclic subspace codes}

\author{
Ismael Gutierrez Garc\'ia$^{1}$ and Ivan Molina Naizir$^{2}$\\ 
\small{Department of Mathematics and Statistics}\\
\small{Universidad del Norte - Barranquilla, Colombia}\\
\small{isgutier@uninorte.edu.co$^{1}$}\\
\small{inaizir@uninorte.edu.co$^{2}$}}

\begin{document}
\maketitle

\begin{abstract}
Construction of subspace codes with good parameters is one of the most important problems in random network coding. In this paper we present first a generalization of the concept of cyclic subspaces codes and further we show that the usual methods for constructing cyclic subspace codes over finite fields works for $m$-quasi cyclic codes, namely the subspaces polynomials and Frobenius mappings.
\end{abstract}


\section{Introduction}
Random Network coding is a new research area in information theory that have interesting applications in practical networking systems, like peer-to-peer content distribution network, bidirectional traffic in a wireless network, residential wireless mesh networks, Ad-hoc sensor networks, and others \cite{Fragouli}. 

In the seminal work \cite{Koetter1}, R. K\"otter and F. R. Kschischang have developed the theory of subspace codes for applications in network coding, 
where it was proved that constant dimension subspaces codes can be used for detection and correction of errors during a packet transmission in a network. 
Similar as in classical coding theory, there are two main branches of research
in random network coding: the existence and construction of optimal codes and the  design and implementation of efficient encoding and decoding algorithms for a given network code.

Cyclic subspace codes were first presented by A. Kohnert and S. Kurz in \cite{Kurz} from the point of view of the theory of designs over finite fields. Later T. Etzion and A. Vardy in \cite{Etzion2} have defined them as a $q$-version of cyclic code from the classical coding theory. Recently T. Etzion et al. have presented in \cite{Etzion1} new methods for constructing such codes, which includes linearized polynomials, namely subspaces polynomials and  Frobenius mappings. J. Rosenthal et al. \cite{Rosenthal} and H. Gluesing et al. \cite{Gluesing} studied cyclic codes from the point of view of groups actions. Specifically, they have used an action of the general linear group $\mathrm{GL}(n,q)$ over a Grassmannian $G_q(n,k)$ to define them. These codes were called cyclic orbits codes. It can easily see that cyclic subspaces codes are a special case of orbits codes. 

In this paper we present the definition of $m$-quasi cyclic subspace codes as a natural generalization of cyclic codes and we show that the techniques used for the construction of cyclic codes in \cite{Etzion1} also works for $m$-quasi cyclic codes. 

\section{Preliminaries}
Let $\F_q^n$ be the $n$-dimensional vector space over  the finite field, with $q$ elements, $\F_q$ (where $q$ is a prime power). We denote with $\P_q(n)$ the projective space of order $n$, that is, the set of all subspaces of $\F_q^n$, including the null space and $\F_q^n$ itself. For a fixed natural number $k$, with $0\leq k\leq n$ we denote with $G_q(n,k)$ the set of all subspaces of $\F_q^n$ of dimension $k$ and we call it the $k$-Grassmannian over $\F_q$ or Grassmannian in short. Then we have
\[\P_q(n) = \bigcup_{k=0}^n G_q(n,k).\]
The cardinality of the set $G_q(n,k)$ is given by the $q$-ary gaussian coefficient ${n\brack k}_q$. It is well known that
\[{n\brack k}_q = \frac{(q^n-1)(q^{n-1}-1)\cdots (q^{n-k+1}-1)}{(q^k-1)(q^{k-1}-1)\cdots (q-1)}.\]
Let $\Z_+$ denote the set of non-negative integers. The set $\P_q(n)$ endowed with the distance $d : \P_q(n)\times \P_q(n) \longrightarrow \Z_+$ defined by
\begin{align*}
d(U, V) &= \dim(U + V) - \dim(U\cap V)\\
           &= \dim(U) + \dim(V) - 2 \dim(U\cap V)
\end{align*}
is a metric space. This distance $d$ is called the \texttt{subspace distance}. A \texttt{subspace code} $\codeC$ is a non empty subset of $\P_q(n)$ and as usually the elements of $\codeC$ are called \texttt{codewords}. Constant dimension codes in network coding are the analogues of constant weight codes in classical coding theory. A \texttt{constant dimension code} or \texttt{grassmannian code} $\codeC$ is just a non empty subset of $G_q(n,k)$. The \texttt{minimum distance} $d(\codeC)$ of a subspace code $\codeC$ is defined as usually, that is, as the smallest distance between any two different codewords.

Let $\codeC$ be a subspace code of minimum distance $d$. Then we say that $\codeC$ is a $[n,|\codeC|,d]$-code over $\F_q$ and $[n,|\codeC|,d]$ are its parameters. If $\codeC$ is a grassmannian code and it has minimum distance $d$, then we say that $\codeC$ is a $[n,k,|\codeC|,d]$-code over $\F_q$ and its parameters are given by $[n,k,|\codeC|,d]$. Notice that in this case, if $U, V\in \codeC$, then 
\[d(U,V) = 2k-2\dim(U\cap V).\] 
Thus $d(\codeC)$ is always an even number.

Let $\mathcal{A}_q(n,d)$, respectively $\mathcal{A}_q(n,d,k)$, denote the maximum number of codewords in an $[n,|\codeC|,d]$-code in $\P_q(n)$, respectively $[n,k,|\codeC|,d]$-code in $G_q(n,k)$. Due the minimum distance for a constant dimension code is always an even number, it suffices to consider $\mathcal{A}_q(n,d,k)$ for $d=2\delta$, for some natural number $\delta$. 
T. Etzion and A. Vardy established in \cite{Etzion2} the following bound:
\begin{equation}\label{bound1}
\mathcal{A}_q(n,2 \delta+2,k)\leq  \frac{{n \brack k}_q}{{n-k+\delta \brack \delta}_q}.
\end{equation}

Let $\F_{q^n}$ be the extension field of $\F_q$ (of degree $n$). It is well known that we may regard $\F_{q^n}$ as a vector space of dimension $n$ over $\F_q$. That is, for a fixed basis, we can identifier every element of $\F_{q^n}$ with a $n$-tuple of elements in $\F_q$. Let $\gamma$ be a primitive element of $\F_{q^n}$. A subspace code $\codeC\subseteq \P_q(n)$ is called \texttt{cyclic}, if it has the following property: 
\[\{0,\gamma^{i_1},\gamma^{i_2},\ldots, \gamma^{i_s}\}\in \codeC \Rightarrow \{0,\gamma^{i_1+1},\gamma^{i_2+1}, \ldots, \gamma^{i_s+1}\}\in \codeC.\]
(Assuming that $s=q^k-1$, with $k$ the dimension of the codeword).

For example, let $\gamma$ be a primitive root of $x^8+x^4+x^3+x^2+1$ and use this polynomial to generate the field $\F_{2^{8}}$. Let $\codeC$ be the constant dimension code in $G_2(8,3)$, which consists of all cyclic shifts of
	\begin{align*}
		& \{ \gamma^{0}, \gamma^{52}, \gamma^{71}, \gamma^{109}, \gamma^{135}, \gamma^{141}, \gamma^{144}\}\\
		& \{ \gamma^{0}, \gamma^{31}, \gamma^{45}, \gamma^{65}, \gamma^{87}, \gamma^{162}, \gamma^{167}\}\\
		& \{ \gamma^{0}, \gamma^{62}, \gamma^{69}, \gamma^{79}, \gamma^{90}, \gamma^{130}, \gamma^{174}\}\\
		& \{ \gamma^{0}, \gamma^{58}, \gamma^{60}, \gamma^{107}, \gamma^{108}, \gamma^{132}, \gamma^{161}\}\\
		& \{ \gamma^{0}, \gamma^{16}, \gamma^{46}, \gamma^{59}, \gamma^{82}, \gamma^{137}, \gamma^{145}\}. 
	\end{align*}
Then $\codeC$ is a $[8,1275,4,3]$-code and using \eqref{bound1} holds:
\[1275 \leq \mathcal{A}_2(8,4,3) \leq 1542.\]
It follows that $\codeC$ is optimal among cyclic codes. In \cite{Etzion2}, \cite{Kurz} and recently in \cite{Gutierrez-Molina} there were found several examples of cyclic subspaces codes with small parameters. 

Given a binary cyclic code $\codeC\subseteq G_2(n,k)$ and a $V\in \codeC$, we associate the corresponding binary \texttt{characteristic vector} $x_V = (x_0, x_1, \ldots, x_{2^k-2})$ as follows: 
\[x_j = \begin{cases}
1 & \text{if} \ \gamma^j \in V\\
0 & \text{if} \ \gamma^j \notin V.
\end{cases}\]
Then the set of all such characteristic vectors is closed under cyclic shifts. Note that the property of being cyclic does not depend on the choice of a primitive element in $\F_q$. This concept is an useful tool, for example, to calculate easily the intersection of two subspaces and of course its dimension.

Let $\alpha\in \F_{q^n}^\ast$ and $V\in G_q(n,k)$. The \texttt{cyclic shift} or the \texttt{orbit} of $V$ is defined as follows:
\[\alpha V := \{\alpha v\mid v\in V\}.\] 
Then a subspace code $\codeC\subseteq G_q(n,k)$ is called \texttt{cyclic}, if for all $\alpha\in \F_{q^n}^\ast$ and all subspace $V\in \codeC$ we have $\alpha V\in \codeC$. That is, $\{\alpha V\mid \alpha\in \F_{q^n}^\ast\}\subseteq \codeC$, for all $V\in \codeC$. It is known that, if $V\in G_q(n,k)$, then 
\begin{equation}\label{orbita}
|\{\alpha V\mid \alpha\in \F_{q^n}^\ast\}| = \tfrac{q^n-1}{q^t-1},
\end{equation}
for some natural number $t$, which divides $n$. As an immediate consequence we have that the maximum size of an orbit is reached when $t = 1$. In this case we say that $V\in G_q(n,k)$ has a \texttt{full length orbit}. Otherwise $V$ has a \texttt{degenerate orbit}. It is clear that the set $\alpha V$ is again a subspace with the same dimension as $V$. If for $\alpha, \beta\in \F_{q^n}^\ast$ holds that $\alpha V\neq \beta V$, then we say that these cyclic shifts are \texttt{distinct}. 

Let $\gamma$ be a primitive element of $\F_{q^n}$ and $m$ a natural number with $m\mid (q^n-1)$. 
A subspace code $\codeC$ is called $m$-\texttt{quasi cyclic}, if holds the following property: 
\[\{0,\gamma^{i_1},\gamma^{i_2},\ldots, \gamma^{i_k}\}\in \codeC \Rightarrow \{0,\gamma^{i_1+m},\gamma^{i_2+m}, \ldots, \alpha^{i_k+m}\}\in \codeC.\]

Let $V\in G_q(n,k)$. The \texttt{$m$-quasi cyclic shift} or the \texttt{$m$-quasi orbit} of $V$ is defined by
\[\alpha^m V := \{\alpha^m v\mid v\in V\}.\] 
Then a constant dimension subspace code $\codeC$ is called $m$-\texttt{quasi cyclic}, if for all $\alpha\in \F_{q^n}^\ast$ and all subspace 
$V\in \codeC$ we have $\alpha^m V\in \codeC$. That is, $\{\alpha^m V\mid \alpha\in \F_{q^n}^\ast\}\subseteq \codeC$, for all $V\in \codeC$. 

The rest of this paper is organized in three sections. In section 3 we consider subspaces polynomial, Frobenius mappings and their connection with $m$-quasi cyclic shift of a given subspace. In section 4 we present a generalization of cyclic subspaces codes, namely the $m$-quasi cyclic codes and we show some properties of them. Conclusions and future works are presented in section fifth.

\section{Subspace Polynomials and Frobenius mappings}

\begin{definition}
	A polynomial over $\F_{q^n}$ of the form
	\[L(x) = \sum_{j=0}^k a_j x^{q^j} \]
	is usually called \texttt{linearized polynomial}.
\end{definition}

Such special kinds of polynomials  were firstly studied by O. Ore in \cite{Ore}. They play an important role in classic coding theory \cite[Chapter 4]{MacWilliams}, 
in addition, over the past decade it has been important also in random network coding \cite{Koetter1}, \cite{Koetter2}.

If in a context $q$ is fixed, then we use $[j]$ to denote $q^j$. In this notation, a linearized polynomial over the extension field $\F_{q^n}$ can be written as
\[L(x) = \sum_{j=0}^k a_j x^{[j]}.\]

\begin{theorem}\cite[Theorem 3.50]{Lidl1}\label{3.50}
	Let $L(x)$ be a nonzero linearized polynomial over $\F_{q^n}$ and let the extension field $\F_{q^r}$ of $\F_{q^n}$ containing all the roots of $L(x)$. 
	Then each root of $L(x)$ has the same multiplicity, which is either one or a power of $q$, and the roots form a linear subspace of 
	$\F_{q^r}$, where $\F_{q^r}$ is regarded as a vector space over $\F_q$.
\end{theorem}

Next theorem is also a partial converse of the previous theorem, which will be used many times in this work.

\begin{theorem}\cite[Theorem 3.52]{Lidl1}\label{3.52}
	Let $V$ be a linear subspace of $\F_{q^n}$ considered as a vector space over $\F_q$. Then for any nonnegative integer $k$ the polynomial 
	\[L(x) := \prod_{v\in V} (x-v)^{[k]}\]
	is a linearized polynomial over $\F_{q^n}$.
\end{theorem}

Taken $[k]=1$ in the previous theorem, then we have the definition of subspaces polynomials. This was  presented by B. Sasson et al. in \cite{Sasson}.

\begin{definition}\label{definicion 2}
A monic linearized  polynomial $L$ over the field $\F_{q^n}$ is  called a \texttt{subspace polynomial}, with respect to $\F_{q^n}$, if and only if $L$ has the form
\[L(x) = \prod_{v\in V} (x-v),\]
for some subspace $V$ in $G_q(n,k)$.
\end{definition}

\begin{remark}\label{definicion 2.1}
It is clear from de definition that the following statements are equivalent:
\begin{enumerate}
	\item[($a$)] $L$ is a subspace polynomial, with respect to $\F_{q^n}$.
	\item[($b$)] $L$ splits completely over $\F_{q^n}$ and all its roots are simple, i.e. they have multiplicity 1.
	\item[($c$)] $L$ divides $x^{[n]}-x$.
\end{enumerate}
\end{remark}

As a consequence of the definition of subspace polynomial and the theorem \ref{3.50} follows the next affirmation.

\begin{lemma}\label{lema 1}
	Let $L$ be a subspace polynomial. Then the coefficient of $x$ is non-zero. Conversely, every linearized polynomial with non-zero coefficient of $x$ is a subspace polynomial in its splitting field.
\end{lemma}  

\begin{remark}\label{lema 2}
	From definition \ref{definicion 2} follows that for a given vector space $V$ of dimension $k$ over $\F_q$ the polynomial $L(x)=\prod_{v\in V} (x-v)$ is the unique subspace polynomial whose roots are the set $V$. It was proved en \cite{Etzion1} that two subspaces are equal if and only if their related subspaces polynomials are equal. Then we can take the notation $L_V$ to refer the subspace polynomial  associated with the subspaces $V$.
\end{remark}

\begin{example}
	Let $\gamma$ be a primitive root of $x^8+x^4+x^3+x^2+1$ and use this polynomial to generate the field $\F_{2^8}$. The followings polynomials over $\F_{2^8}$ are linearized and also subspace polynomials.	
	\begin{enumerate}
		\item $L(x)= x^{[3]} + \gamma^{103}x^{[2]} + \gamma^{74}x$
		\item $L(x) = x^{[4]} + \gamma^{238}x^{[2]} + \gamma^{51}x$
		\item $L(x)= x^{[4]} + \gamma^{251}x^{[3]} + \gamma^{8}x^{[2]} +  \gamma^{182}x^{[1]} + \gamma^{207}x$.
	\end{enumerate}
\end{example}

\begin{lemma}\label{lema 4.2} 
Let $\alpha\in \F_{q^n}^\ast$ and $m$ be a natural number, with $m\mid q^n-1$.
Let $V\in G_q(n,k)$ and $U\in G_q(n,l)$ are two distinct subspaces and 
	\[L_V(x) = x^{[k]} + \sum_{j=0}^t a_j x^{[j]},\]
	\[L_U(x) = x^{[l]} + \sum_{j=0}^s b_j x^{[j]},\] 
with $a_t\neq 0$ and $b_s\neq 0$. Then 
\begin{enumerate}
 \item[$(a)$] If $k\leq l$, then $\dim(\alpha^mV\cap \alpha^mU) \leq \max(s,t+l-k)$.
 \item[$(b)$] If $k=l$, then
 \[\dim(\alpha^mV\cap \alpha^mU) \leq k-\min(k-t,k-s)\]
 and
 \[d(\alpha^mV,\alpha^mU) \geq 2\min(k-t,k-s).\]
\end{enumerate}
\end{lemma}

\begin{proof}
\begin{enumerate}
\item[$(a)$] Using  \cite[Lemma 4]{Etzion1} we have that $\dim(U\cap V) \leq \max(s,t+l-k)$. Therefore, it is sufficient to show that $\dim(U\cap V)= \dim(\alpha^mV\cap \alpha^mU)$, which is immediately. To see that just consider the characteristic vectors of the subspaces $U$, $V$ $\alpha^mU$ and $\alpha^mU$.
	
\item[$(b)$] Suppose $k=l$. From $(a)$ we have that $\dim(\alpha^mV\cap \alpha^mU) \leq \max(s,t)$. On the other hand
\[k-\min(k-t,k-s) = k-(k-\max(t,s)) = \max(t,s).\]
Finally,
\begin{align*}
d(\alpha^mV,\alpha^mU) &= \dim(\alpha^mV)+\dim(\alpha^mU)-2\dim(\alpha^mV\cap \alpha^mU)\\
						&= 2k-2\dim(\alpha^mV\cap \alpha^mU)\\
						&= 2(k-\dim(\alpha^mV\cap \alpha^mU))\\
						&\geq 2\min(k-t,k-s).
\end{align*}
\end{enumerate}
\end{proof}

By an automorphism $\sigma$ of $\F_{q^n}$ over $\F_q$ we mean a automorphism of $ \F_{q^n}$ that fixes the elements of $\F_q$. That is, $\sigma$ is a bijective function on $\F_{q^n}$,  $\sigma(x+y) = \sigma(x) + \sigma(y)$, $\sigma(xy) = \sigma(x) \sigma(y)$, for all $x, y\in \F_{q^n}$ and $\sigma(x) = x$, for all $x\in \F_q$ (see \cite{Lidl1}, page 49). These functions are also called \texttt{Frobenius mappings}.

\begin{theorem}\cite[Theorem 2.21]{Lidl1}
	The distinct automorphisms of $\F_{q^n}$ over $\F_q$ are exactly the mappings $\sigma_j : \F_{q^n}\longrightarrow \F_{q^n}$, defined by 
	\[\sigma_j(x) = x^{[j]},\] 
	where $x\in \F_{q^n}$ and $j\in \{0,1,\ldots, n-1\}$.
\end{theorem}

\begin{definition}
	Let $V\in G_q(n,k)$ and $j\in \{0,1,\ldots, n-1\}$. The $i$-th \texttt{Frobenius shift} of $V$ is defined as the image of $V$ under $\sigma_j$. That is,
	\[\sigma_j(V):= \{\sigma_j(v)\mid v\in V\}.\] 
\end{definition}

Since $\sigma_j$ is an automorphism, we have that $\sigma_j(V)\in G_q(n,k)$. 
Now we investigate the relationship between $m$-quasi cyclic subspaces codes and Frobenius mappings. Next Lemma shows how is the subspace polynomial of the subspace resulting by applying the $i$-th Frobenius mapping. It appears as \cite[Lemma 6]{Etzion1}. We include the proof for the sake of completeness.

\begin{lemma}\label{lema 6} 
If $V\in G_q(n,k)$ and $L_V(x)= x^{[k]} + \sum_{j=0}^i a_jx^{[j]}$, then for all 
	$s\in \{0,1,\ldots, n-1\}$ holds
	\[L_{\sigma_s(V)}(x)= x^{[k]} + \sum_{j=0}^i \sigma_s(a_j)x^{[j]}.\]
\end{lemma}

\begin{proof} 
Let $s\in \{0,1,\ldots, n-1\}$ fixed and $y\in \sigma_s(V)$. Then there exist $v\in V$ such that $y=\sigma_s(v)$. Since $\sigma_j$ is an automorphism we have
	
	\begin{align*}
	y^{[k]} + \sum_{j=0}^i \sigma_s(a_j)y^{[j]} & = \sigma_s(v)^{[k]} + \sum_{j=0}^i \sigma_s(a_j) \sigma_s(v)^{[j]}\\ 
	&= \sigma_s(v^{[k]}) + \sum_{j=0}^i \sigma_s(a_jv^{[j]}) \\
	&= \sigma_s\bigg(v^{[k]} + \sum_{j=0}^i a_jv^{[j]}\bigg) \\
	&= \sigma_s\big(L_V(v)\bigg) \\
	&= \sigma_s\bigg(\prod_{w\in V} (v-w)\bigg) \\
	&= \sigma_s(0) \\
	&=0.
	\end{align*}
This proves that every element of $\sigma_s(V)$ is a root of the polynomial $x^{[k]} + \sum_{j=0}^i \sigma_s(a_j) x^{[j]}$. Due to the degree of this polynomial, it must be the subspace polynomial $L_{\sigma_s(V)}(x)$.
\end{proof}

Let $m\in \N$, with $m\mid q^n-1$. Now let us denote with $(\F_{q^n})^m$ the set of all $m$-powers of elements of $\F_{q^n}$. That is, 
\[(\F_{q^n})^m =\{\alpha^m\mid \alpha\in \F_{q^n}\}.\]
This set has $\frac{1}{m}\big(q^n-1\big)$ elements.

For $\alpha, \beta\in \F_{q^n}^\ast$, and an integer number $t$, with $t$ divides $n$, we define on $(\F_{q^n})^m$ the following relation $\sim_t$:
\begin{equation}\label{equivalencia1}
\alpha^m \sim_t \beta^m :\Longleftrightarrow \frac{\alpha^m}{\beta^m}\in \F_{q^t}^\ast.
\end{equation}

It is immediate that $\sim_t$ define an equivalence relation on $(\F_{q^n})^m$ and that the equivalence class $[\alpha^m]$ of $\alpha^m$ under this relation is the $m$-quasi cyclic shift $\alpha^m \F_{q^t}^\ast$ of $\F_{q^t}^\ast$.  Therefore, there are exactly $\frac{1}{m}\big(\frac{q^n-1}{q^t-1}\big)$ equivalence classes of $\sim_t$, each of which has $q^t-1$ elements.

Next Lemma shows how is the subspace polynomial of the subspace resulting by applying a $m$-quasi cyclic shift.

\begin{lemma}\label{lema 5}
Let $\alpha\in \F_{q^n}^\ast$ and $m$ be a natural number, with $m\mid q^n-1$ and 
	$V\in G_q(n,k)$. Then 
	\[L_{\alpha^m V}(x) = \alpha^{m[k]} L_V(\alpha^{-m}x).\]
Furthermore, if $L_V(x)= x^{[k]} + \sum_{j=0}^i a_jx^{[j]}$, then 
	\[L_{\alpha^m V}(x) = x^{[k]} + \sum_{j=0}^i \alpha^{m([k]-[j])}a_jx^{[j]}.\]
\end{lemma}

\begin{proof}
	Using the definition we have,
	\begin{align*}
	L_{\alpha^m V}(x) &= \prod_{u\in \alpha^m V} (x-u)\\
	&= \prod_{v\in V} (x-\alpha^m v)\\
	&= \prod_{v\in V} \alpha^m(\alpha^{-m}x - v)\\
	&= \alpha^{m[k]}\prod_{v\in V} (\alpha^{-m}x - v)\\
	&= \alpha^{m[k]} P_V(\alpha^{-m}x).
	\end{align*}
	If $L_V(x)= x^{[k]} + \sum_{j=0}^i a_jx^{[j]}$, then 
	\begin{align*}
	L_{\alpha^m V}(x) &= \alpha^{m[k]} \bigg( (\alpha^{-m}x)^{[k]} + \sum_{j=0}^i a_j (\alpha^{-m}x)^{[j]} \bigg)\\ 
	&= \alpha^{m[k]}\alpha^{-m[k]} x^{[k]} + \sum_{j=0}^i a_j \alpha^{m[k]} \alpha^{-m[j]} x^{[j]}\\
	&= x^{[k]} + \sum_{j=0}^i \alpha^{m([k]-[j])}a_jx^{[j]}.
	\end{align*}
\end{proof}

In the following result we present a connection between the coefficients of the polynomial $P_{\alpha^m V}(x)$ of a given subspace $V\in G_q(n,k)$ and the number of its different $m$-quasi cyclic shifts. 

\begin{lemma}\label{lema 7} 
Let $V\in G_q(n,k)$ and $L_V(x)= x^{[k]} + \sum_{j=0}^i a_jx^{[j]}$. If $a_s\neq 0$, for some $s\in \{1,\ldots, i\}$ and $t:=\gcd(s,n)$, then
\[\left|\{\alpha^m V\mid \alpha\in \F_{q^n}^\ast\}\right| \geq \frac{1}{m}\left(\frac{q^n-1}{q^t-1}\right).\]
\end{lemma}

\begin{proof}
Assume that $\alpha^m V = \beta^m V$, where $\alpha, \beta\in \F_{q^n}^\ast$. By Lemma \ref{lema 5} we have 
\[L_{\alpha^m V}(x) = x^{[k]} + \sum_{j=0}^i \alpha^{m([k]-[j])}a_jx^{[j]}\]
and	
\[L_{\beta^m V}(x) = x^{[k]} + \sum_{j=0}^i \beta^{m([k]-[j])}a_jx^{[j]}.\]
Due to $\alpha^m V = \beta^m V$ and the uniqueness of the subspace  polynomial (Remark \ref{lema 2}) follows
\[\sum_{j=0}^i \alpha^{m([k]-[j])}a_jx^{[j]} = \sum_{j=0}^i \beta^{m([k]-[j])}a_jx^{[j]}.\]
Then
\begin{align*}
a_s \alpha^{m([k]-[s])} &= a_s\beta^{m([k]-[s])}\\
a_0 \alpha^{m([k]-1)} &= a_0\beta^{m([k]-1)}.
\end{align*}
If $a_s\neq 0$ and since $a_0\neq 0$ (due to Lemma \ref{lema 1}) follows that
	\begin{align}
	\bigg(\frac{\alpha^m}{\beta^m}\bigg)^{[k]-[s]} & = 1 \label{A1}\\
	\bigg(\frac{\alpha^m}{\beta^m}\bigg)^{[k]-1} &= 1. \label{A2}
	\end{align}
By dividing \eqref{A2} by \eqref{A1} follows that
	\[\bigg(\frac{\alpha^m}{\beta^m}\bigg)^{[s]-1} = 1.\]
Then
\[\ord\bigg(\frac{\alpha^m}{\beta^m}\bigg) \ \bigg|\ \gcd(q^n-1,q^s-1)\]
Using \cite{Knuth} (Exercise 38 on page 147)  we have
	\[\gcd(q^n-1,q^s-1) = q^{\gcd(n,s)-1}.\]
Therefore
\[\ord\bigg(\frac{\alpha^m}{\beta^m}\bigg) \ \bigg|\ q^{\gcd(n,s)-1},\]
which implies $\frac{\alpha^m}{\beta^m}\in \F_{q^{\gcd(n,s)}}^\ast = \F_{q^t}^ \ast$, consequently it follows that $\alpha^m \sim_t \beta^m$. We know that there exist exactly  $\frac{1}{m}\big(\frac{q^n-1}{q^t-1}\big)$ equivalence classes of $\sim_t$, which implies the assertion. 
\end{proof}

\begin{remark}
Let	$\alpha, \beta\in \F_{q^n}^\ast$, $V\in G_q(n,k)$ and $L_V(x)= x^{[k]} + \sum_{j=0}^i a_jx^{[j]}$. If $a_s\neq 0$, for some $s\in \{1,\ldots, i\}$ and $t:= \gcd(s,n)=t$, then 
\begin{enumerate}
\item[$(a)$] If $\alpha^m V = \beta^m V$, then $\alpha^m \sim_t \beta^m$
\item[$(b)$] If $\alpha^m \not\sim_t \beta^m$, then $\alpha^m V\neq \beta^m V$.
\end{enumerate}
\end{remark}

In the construction of $m$-quasi cyclic subspace codes play an important role a special class of linearized polynomials. These are subspaces polynomials, which in turn are certain Trinomials.

\begin{lemma}\label{lema 8} 
Let $m$ be a natural number, with $m\mid q^n-1$, $V\in G_q(n,k)$ and $L_V(x)= x^{[k]} + a_sx^{[s]} + a_0x$, where $a_s\neq 0$. If there exists $\alpha\in \F_{q^n}^\ast$ and $i\in \{0,1,\ldots, n-1\}$ such that $\sigma_i(V)=\alpha^mV$, then 
\[\left(\frac{a_0^{\frac{q^k-q^s}{q^s-1}}}{a_1^{\frac{q^k-1}{q^s-1}}}\right)^{q^i-1} = 1.\]
\end{lemma}

\begin{proof}
Assume that $\sigma_i(V)=\alpha^mV$, for some $i\in \{0,1,\ldots, n-1\}$ and some $\alpha\in \F_{q^n}^\ast$. Using Lemmas \ref{lema 6} and \ref{lema 5} follows
\begin{align*}
L_{\alpha^m V}(x) &= x^{[k]} + a_s \alpha^{m([k]-[s])} x^{[s]} + a_0 \alpha^{m([k]-1)} x\\ 
	L_{\sigma_i(V)}(x) &= x^{[k]} + \sigma_i(a_s)xc{[s]} + \sigma_i(a_0) x.
	\end{align*}
	By Remark \ref{lema 2},
	\begin{align*}
	a_s \alpha^{m([k]-[s])} & = \sigma_i(a_s)  = a_s^{[i]}\\
	a_0 \alpha^{m([k]-1)} &= \sigma_i(a_0)  =a_0^{[i]}.
	\end{align*}
	Since $a_0\neq 0$ and $a_s\neq 0$, we have
	\begin{align}
	\alpha^{m([k]-[s])} & = a_s^{[i]-1} \label{A3}\\
	\alpha^{m([k]-1)} &= a_0^{[i]-1}. \label{A4}
	\end{align}
Therefore
	\[\alpha^{m([s]-1)} = \left(\frac{a_0}{a_s}\right)^{[i]-1}.\]
That is,
	\begin{equation}\label{A5}
	\alpha^{m(q^s-1)} = \left(\frac{a_0}{a_s}\right)^{q^i-1}.
	\end{equation}
From \eqref{A3} it follows 
	\[\alpha^{m\left(\frac{q^k-q^s}{q^s-1}\right)(q^s-1)} =  a_s^{q^i-1}.\]
Using \eqref{A5} we have
	\[\left(\frac{a_0}{a_s}\right)^{(q^i-1)\left(\frac{q^k-q^s}{q^s-1}\right)} =  a_s^{q^i-1}.\]
Thus
	\[\frac{a_0^{\left(\frac{q^k-q^s}{q^s-1}\right)(q^i-1)}}{a_s^{\left(\frac{q^k-q^s}{q^s-1}\right)(q^i-1)+(q^i-1)}} =1.\]
Finally
	\[\frac{a_0^{\left(\frac{q^k-q^s}{q^s-1}\right)(q^i-1)}}{a_s^{\left(\frac{q^k-1}{q^s-1}\right)(q^i-1)}} =1,\]
	which implies  
\[\left(\frac{a_0^{\frac{q^k-q^s}{q^s-1}}}{a_1^{\frac{q^k-1}{q^s-1}}}\right)^{q^i-1} = 1.\]
\end{proof}

\section{$m$-quasi cyclic subspaces codes}

\begin{definition}
Let $\alpha$ be a primitive element of $\F_{q^n}$ and $m$ a natural number with $m\mid (q^n-1)$. 
A subspace code $\codeC\subseteq \P_q(n)$ is called $m$-\texttt{quasi cyclic}, if holds the following property: 
\[\{0,\alpha^{i_1},\alpha^{i_2},\ldots, \alpha^{i_k}\}\in \codeC \Rightarrow \{0,\alpha^{i_1+m},\alpha^{i_2+m}, \ldots, \alpha^{i_k+m}\}\in \codeC.\]
\end{definition}

Now we present a natural generalization for the definition of orbit of a subspace and for the length of an orbit.

\begin{definition}
Let $\alpha\in \F_{q^n}^\ast$, $m$ a natural number with $m\mid q^n-1$ and $V\in G_q(n,k)$. 
The \texttt{$m$-quasi cyclic shift} or the \texttt{$m$-quasi orbit} of $V$ is defined by
\[\alpha^m V := \{\alpha^m v\mid v\in V\}.\] 
Then a code $\codeC\subseteq G_q(n,k)$ is called $m$-\texttt{quasi cyclic}, if for all $\alpha\in \F_{q^n}^\ast$ and all subspace 
$V\in \codeC$ we have $\alpha^m V\in \codeC$. That is, $\{\alpha^m V\mid \alpha\in \F_{q^n}^\ast\}\subseteq \codeC$, for all $V\in \codeC$.
\end{definition}

\subsection{$m$-quasi cyclic codes with a single quasi orbit}

The following Lemma, whose proof is inspired in the idea presented by T. Etzion at al. in \cite[Lemma 9]{Etzion1} for the case $m=1$. 
The demonstration is obtained only by performing basic modifications to the cited one.

\begin{lemma}\label{lema9}
If $m$ a natural number with $m\mid q^n-1$ and $V\in G_q(n,k)$, then \[|\{\alpha^m V\mid \alpha\in \F_{q^n}^\ast\}| = \frac{1}{m}\left(\frac{q^n-1}{q^t-1}\right),\] 
for some natural number $t$, which divides $n$.
\end{lemma}

\begin{proof}
Let $\gamma$ be a primitive element in $\F_{q^n}$, that is, $\F_{q^n}^\ast =\langle \gamma\rangle$ and let $l$ the smallest natural number with $\gamma^{lm}V=V$.
It is clear that $lm\mid q^n-1$. Let now $0\leq s < l$ and $i\in \N$, then
\begin{align*}
\gamma^{iml+s}V & = \gamma^s(\gamma^{iml}V)\\
&= \gamma^s(\gamma^{ml}\cdots \gamma^{ml})V\\
&= \gamma^sV.
\end{align*}
That is, for each natural number $i$ and for each $0\leq s < l$ is verified that $\gamma^sV= \gamma^{iml+s}V$.
Additionally, for every $0\leq s_1, s_2< l$ the sets 
\[A_{s_j} := \{\gamma^{iml+s_j}\mid i\in \N\}\]
satisfy that $|A_{s_1}| = |A_{s_2}|$. In fact, given that $q^n-1 = wml$, for some $w\in \N$, then we have
\[A_{s_j} = \{\gamma^{s_j}, \gamma^{ml+s_j},\ldots, \gamma^{ml(w-1)+s_j}\}.\]
Therefore $|A_{s_1}| = |A_{s_2}|=w$. Let $\gamma^{iml}, \gamma^{rml}\in A_0$, for some $i, r\in \N$. Since $A_0 = \{\gamma^{iml}\mid i\in \N\}$, it follows that
\[(\gamma^{iml} + \gamma^{rml})V \subseteq \gamma^{iml}V + \gamma^{rml}V = V+V = V,\]
and therefore $\gamma^{iml} + \gamma^{rml}\in A_0$. It is clear that $A_0$ is closed under multiplication, then we have that $\langle \gamma^{ml}\rangle$ is the multiplicative group of a subfield $\F_{q^t}$ of $\F_{q^n}$, where $t$ is a natural number, which divides $n$. Then
\[|\{\alpha^m V\mid \alpha\in \F_{q^n}^\ast\}| = l= \tfrac{q^n-1}{mw} =\tfrac{1}{m}\big(\tfrac{q^n-1}{q^t-1}\big),\]
which proves affirmation.
\end{proof}

An immediate consequence of Lemma \ref{lema9} is that the largest possible size of an $m$-quasi orbit is $\frac{1}{m}\big(\frac{q^n-1}{q-1}\big)$. This justifies the following definition:

\begin{definition}
We say that the subspace $V\in G_q(n,k)$ has a \texttt{full length $m$-quasi orbit}, if
\[|\{\alpha^m V\mid \alpha\in \F_{q^n}^\ast\}| = \tfrac{1}{m}\big(\tfrac{q^n-1}{q-1}\big).\]
In other case we say that it has a \texttt{degenerate $m$-quasi orbit}. 
\end{definition}

It is clear that the set $\alpha^m V$ is again a subspace with the same dimension as $V$.  

\begin{theorem}\cite[Chapter 4, Theorem 10]{MacWilliams}\label{lema 10}
The polynomial $x^{[n]}-x$ is the product of all monic polynomials, irreducible 
over $\F_q$, whose degree divides $n$. 	
\end{theorem}

In the following we assume that $k,s\in \N$, with $s<k$.

\begin{theorem}
If $q^k-1$ divides $n$ and the polynomial $x^{[k]-1}+ x^{[s]-1} + 1$ is irreducible over $\F_q$, then the polynomial $x^{[k]}+ x^{[s]} + x$ is a subspace polynomial with respect to $\F_{q^n}$.
\end{theorem}

\begin{proof}
Assume that the polynomial $f(x)= x^{[k]-1}+ x^{[s]-1} + 1$ is irreducible over $\F_q$. Due to $\deg(f)$ divides $n$, by Theorem \ref{lema 10} we have that $f$ divides $x^{[n]}-x$, and hence $x^{[k]}+ x^{[s]} + x$ divides $x^{[n]}-x$. Thus using Remark \ref{definicion 2.1} follows that $x^{[k]}+x^{[s]}+x$ is a subspace polynomial.	
\end{proof}

\begin{corollary}\label{corolario 4}
If $q^k-1$ divides $n$ and $x^{[k]-1}+ x^{[s]-1} + 1$ is irreducible over $\F_q$, $\gcd(n,s)=t$ and $V\in G_q(n,k)$ is the subspace whose subspace polynomial is $L_V(x)= x^{[k]}+ x^{[s]} + x$, then 
\[\codeC := \{\alpha^m V\mid \alpha\in \F_{q^n}^\ast\}\]
is a $m$-quasi cyclic subspace code with parameters $\left[n,k,|\codeC|,d\right]$, with $|\codeC|\geq \tfrac{1}{m}\big(\tfrac{q^n-1}{q^t-1}\big)$ and $d\geq 2(k-s)$. In particular, if $\gcd(n,s)=1$, then the code $\codeC$ is a full length $m$-quasi orbit.
\end{corollary}

\begin{proof}
The two first parameters are clear. Since $a_s\neq 0$, the third one follows from Lemma \ref{lema 7}. The assertion about the minimum distance follows from Lemmas \ref{lema 4.2} and \ref{lema 5}.
\end{proof}

A natural question is: when these polynomials are irreducible? At the moment does not exist an explicit construction of this class of irreducible Trinomials. 

\begin{example}
All the following Trinomials are irreducible over $\F_2$ and its roots are contained in $\F_{q^N}$. These were found using GAP \cite{GAP}.
\begin{center}
	\begin{tabular}{|c||l||}
		\hline
		$x^{[k]-1}+ x^{[s]-1} + 1$ & \begin{tabular}{c|c|c}
			$k$ & $s$ & $N$\\
			\hline
			3 & 2 & 7\\
			\hline
			3 & 2 & 21\\
			\hline
			4 & 3 & 15\\
			\hline
			4 & 3 & 30\\
			\hline
			5 & 3 & 31\\
			\hline
			6 & 5 & 63\\
			\hline
			7 & 3 & 127\\
			\hline
			7 & 4 & 127\\
			\hline
			7 & 6 & 127\\
		\end{tabular}
		\\
		\hline
	\end{tabular}
\end{center}
Then the polynomial $x^{2^k}+ x^{2^s}+1$ is a subspace polynomial, say $L_V(x)$, for some $V\in G_2((2^k-1)t,k)$ for all natural number $t$. If we define
\[\codeC := \{\alpha^mV\mid \alpha\in \F_{(2^k-1)t}^\ast\}\]
then is $\codeC$ a $m$-quasi cyclic code with parameters \[[(2^k-1)t,k,2^{(2^k-1)t}-1,2(k-s)].\]

\end{example}

\subsection{$m$-quasi cyclic codes with multiple quasi orbits}

Let $N =nt$,  with $n$ a prime number, $m\mid q^n-1$, let $\gamma$ be a primitive element in $\F_{q^N}$ and $V\in G_q(N,k)$. The set 
\[\bigg\{\gamma^{i\frac{(q^N-1)}{(q^n-1)}} \ \bigg| \ \ i= 0, 1, \ldots, q^{n-2}\bigg\}\cup \big\{0\big\}\]
is the unique subfield  $\F_{q^n}$ of  $\F_{q^N}$. Since  $\F_{q^n}\subseteq \F_{q^N}$ we can view $V$ as a subspace of  $\F_{q^N}$ over  $\F_q$.

Now, we show that the general method introduced by Etzion et al. in \cite{Etzion1} for constructing constant dimension cyclic codes with more than one full length orbit, works exactly for $m$-quasi cyclic subspaces codes. The construction method is by using Frobenius mappings.

\begin{lemma}\label{lema 12}
Let $N=nt$, with $n$ a prime number, $V\in G_q(N,k)$ and $L_V= x^{[k]}+ a_sx^{[s]} + a_0x$, where $a_0, a_s\in \F_{q^n}^\ast$. If $a_s^{\frac{q^k-1}{q^s-1}} \not\sim_1 a_0^{\frac{q^k-q^s}{q^s-1}}$, then the code $\codeC\subseteq G_q(N,k)$ defined by
\begin{equation}\label{codigo-frobenius}
\codeC := \bigcup_{i=0}^{n-1} \{\alpha^m \sigma_i(V)\mid \alpha\in \F_{q^N}^\ast\}
\end{equation}
is a $m$-quasi cyclic $\left[n,k,\frac{n}{m}\big(\tfrac{q^N-1}{q-1}\big),2(k-s)\right]$-code.
\end{lemma}
 
\begin{proof}
From the definition it is clear that $\codeC$ is a $m$-quasi cyclic subspaces code. Using Lemmas \ref{lema 4.2}, \ref{lema 5} and \ref{lema 6} follows that the dimension of the intersection of any two distinct subspaces in $\codeC$ is at most $s$, therefore the minimum distance of $\codeC$ is $2(k-s)$.

To prove the statement about the codes size, we fix $i\in \{0, 1,\ldots, n-1\}$ and notice that Lemma \ref{lema 6} implies that the coefficient of $x^{[s]}$ in $L_{\sigma_i(V)}(x)$ is non-zero. Then by Lemma \ref{lema 7} we have that the set
\[\{\alpha^m \sigma_i(V)\mid \alpha\in \F_{q^N}^\ast\}\]
has $\tfrac{1}{m}\left(\tfrac{q^N-1}{q-1}\right)$ distinct subspaces. 

Finally we just prove that any two different sets in union \eqref{codigo-frobenius} have empty intersection. Let 
$i, j\in \{0, 1,\ldots, n-1\}$, with $i\neq j$ and suppose that there exists $\beta, \delta\in \F_{q^N}^\ast$ such that
\[\beta^m\sigma_i(V) = \delta^m\sigma_j(V).\]
Without loss of generality we can assume that $j>i$ and define $U:= \sigma_i(V)$.
Then by Lemma \ref{lema 6} we have
\begin{align*}
L_U(x) &= L_{\sigma_i(V)}(x)\\
	   &= x^{[k]}+ \sigma_i(a_s)x^{[s]} + \sigma_i(a_0)x\\
	   &= x^{[k]}+ a_s^{[i]}x^{[s]} + a_0^{[i]}x.
\end{align*}
Due to $\sigma_{j-i}(U) = (\frac{\beta}{\delta})^m U$, we can use Lemma \ref{lema 8} to obtain 
\begin{equation}\label{(2)}
\left(\frac{(a_0^{q^i})^{\frac{q^k-q^s}{q^s-1}}}{(a_s^{q^i})^{\frac{q^k-1}{q^s-1}}}\right)^{q^{j-i}-1} = 1.
\end{equation}
Define $z:= \frac{a_0^{\frac{q^k-q^s}{q^s-1}}}{a_s^{\frac{q^k-1}{q^s-1}}}$. Then
\begin{enumerate}
\item[$(1)$] Hypothesis $a_s^{\frac{q^k-1}{q^s-1}} \not\sim_1 a_0^{\frac{q^k-q^s}{q^s-1}}$ implies $z\notin \F_q$.
\item[$(2)$] Equation \eqref{(2)} implies $z^{q^i(q^{j-i}-1)} =1$.
\item[$(3)$] Due to $a_0, a_s\in \F_{q^n}^\ast$, follows that $z\in \F_{q^n}^\ast$.  
\end{enumerate}
Using the same argument as in \cite[Lemma 12]{Etzion1} follows that $\ord(z)\mid q-1$, which implies $z\in \F_q$, a contradiction to $(1)$. This shows that code $\codeC$ is constituted by $n$ sets and the proof is complete.
\end{proof}

Next Lemma shows that the coefficients $a_0$ and $a_s$ from preceding Lemma can found easily in $\F_{q^n}^\ast$.

\begin{lemma}\label{lema 13}
Let $N=nt$, with $n$ a prime number, $m\mid q^n-1$ and let $\gamma$ be a primitive element in $\F_{q^n}$. If $a_0:= \gamma^m$ and $a_s := \gamma^{mq^s}$, then 
\[a_s^{\frac{q^k-1}{q^s-1}} \not\sim_1 a_0^{\frac{q^k-q^s}{q^s-1}}.\]
\end{lemma}

\begin{proof}
Suppose that $a_s^{\frac{q^k-1}{q^s-1}}\sim_1 a_0^{\frac{q^k-q^s}{q^s-1}}$. Than there exist $\alpha\in \F_q^\ast$ such that
\begin{equation}\label{3}
\alpha (\gamma^{mq^s})^{\frac{q^k-1}{q^s-1}} = (\gamma^m)^{\frac{q^k-q^s}{q^s-1}}.
\end{equation}
Raising both sides of the equation \eqref{3} to the $(q^s-1)$-th power we have
\[(\gamma^{mq^s})^{(q^k-1)} = (\gamma^m)^{(q^k-q^s)}.\]
That is,
\begin{equation}\label{4}
\gamma^{mq^k(q^s-1)} = 1.
\end{equation}
Let $u$ and $v$ the multiplicative inverses of $m$ and $q^k$, respectively, modulo $q^n-1$. Then raising again both sides of the equation \eqref{4} to the $uv$-th power follows that $\gamma^{(q^s-1)} = 1$, and therefore $(q^n-1)\mid (q^s-1)$, which implies that $n\mid s$, which in not possible.
\end{proof}

\begin{theorem}\label{theorem 4}
Let $n$ be a prime number, $m\mid q^n-1$ and let $\gamma$ be a primitive element in $\F_{q^n}$. Define $a_0:= \gamma^m$ and $a_s := \gamma^{mq^s}$. If $\F_{q^N}$ is the splitting field of the polynomial $x^{[k]}+ a_sx^{[s]} + a_0x$ and $V\in G_q(N,k)$ its corresponding subspace, then the code $\codeC\subseteq G_q(N,k)$ defined by
\begin{equation}\label{codigo-final}
\codeC := \bigcup_{i=0}^{n-1} \{\alpha^m \sigma_i(V)\mid \alpha\in \F_{q^N}^\ast\}
\end{equation}
is a $m$-quasi cyclic subspaces code with parameters $\left[n,k,\frac{n}{m}\big(\tfrac{q^N-1}{q-1}\big),2(k-s)\right]$.
\end{theorem}

In the following Remark we can see that it is possible to construct $m$-quasi cyclic codes from a given cyclic code $\codeC$.

\begin{remark}
Let $\codeC$ be a cyclic subspace Code with parameters $[n,k,|\codeC|,d]$, say 
\[\codeC = \bigcup_{i\in I} \mathcal{O}_i,\] 
where $O_i$ is an orbit, with $|O_i| = \lambda_i$ and
\[\lambda_i = \left(\frac{q^n-1}{q^{t_i}-1}\right),\] 
for $t_i\in \N$.
If $m\in \N$ divides $\gcd(\lambda_i\mid i\in I)$, then $\codeC$ is a $m$-quasi cyclic subspace Code with parameters $[n,k,|\codeC|,d]$ with 
\[\codeC = \bigcup_{i\in L} \mathcal{Q}_i,\]
where $\mathcal{Q}_i$ is an $m$-quasi orbit and $|L|=m|I|$.
\end{remark}

\section{Conclusions and future work}
In this paper we define the $m$-quasi cyclic subspaces codes as a generalization of cyclic subspace codes. We have proved that the techniques used for constructing cyclic codes, such as subspaces polynomials and Frobenius mappings, can be used to for the new codes. Specially we have showed that the form of the trinomial can be more general as considered in previews works. 

For future investigations, we can consider the generalization of well knows results about cyclic subspaces codes with degenerate orbits and the connection between $m$-quasi cyclic subspaces codes and orbits codes.

\section{Acknowledgements}
The first Author would like to thank the hospitality of the Institute for Algebra and Geometry at Otto von Guericke University - Magdeburg, where part of this work was carried out. Also acknowledges and thanks the financial support of the Deutscher Akademischer Austausch Dienst.

\bibliographystyle{abbrv} 
\bibliography{Isgutier.07.2016}

\begin{thebibliography}{10}

\bibitem{Etzion1}
E.~Ben-Sasson, T.~Etzion, A.~Gabizon, and N.~Raviv.
\newblock Subspace polynomials and cyclic subspace codes.
\newblock {\em arXiv:1404.7739}, 2015.

\bibitem{Sasson}
E.~Ben-Sasson, S.~Kopparty, and J.~Radhakrishnan.
\newblock Subspace polynomials and list decoding of reed-solomon codes.
\newblock {\em IEEE Transactions on Information Theory}, 56:113--120, 2010.

\bibitem{Etzion2}
T.~Etzion and A.~Vardy.
\newblock Error-correcting codes in projective space.
\newblock {\em IEEE Transactions on Information Theory}, 57(2), 2011.

\bibitem{Fragouli}
C.~Fragouli, J.-Y.~L. Boudec, and J.~Widmer.
\newblock Network coding: An instant primer.
\newblock {\em ACM SIGCOMM Computer Communication Review}, 36:63--68, 2006.

\bibitem{GAP}
GAP.
\newblock Groups, algorithms, programming - a system for computational discrete
  algebra.
\newblock \url{http://www.gap-system.org/}.

\bibitem{Gluesing}
H.~Gluesing-Luerssen, K.~Morrison, and C.~Troha.
\newblock Cyclic orbit codes and stabilizer subfields.
\newblock {\em Advances in Mathematics of Communications}, 9(2):177--197, May
  2015.

\bibitem{Gutierrez-Molina}
I.~Gutierrez and I.~Molina.
\newblock Some constructions of cyclic and quasi-cyclic subspaces codes.
\newblock {\em arXiv:1504.04553v4}, 2015.

\bibitem{Koetter1}
R.~Koetter and F.~R. Kschischang.
\newblock Coding for errors and erasures in random network coding.
\newblock {\em IEEE Transactions on Information Theory}, 54:3579 -- 3591, 2008.

\bibitem{Koetter2}
R.~Koetter, D.~Silva, and F.~R. Kschischang.
\newblock A rank-metric approach to error control in random network coding.
\newblock {\em IEEE Transactions on Information Theory}, 54, 2008.

\bibitem{Kurz}
A.~Kohnert and S.~Kurz.
\newblock Construction of large constant dimension codes with a prescribed
  minimum distance.
\newblock {\em Mathematical Methods in Computer Science. Lecture Notes in
  Computer Science}, 5393:31--42, 2008.

\bibitem{Lidl1}
R.~Lidl and H.~Niederreiter.
\newblock {\em {Introduction to Finite Fields and their Applications}}.
\newblock Cambridge University Press, 2 edition, 1994.

\bibitem{MacWilliams}
F.~J. MacWilliams and N.~Sloane.
\newblock {\em {The Theory of Error-Correcting Codes}}.
\newblock North-Holland Mathematical Library, Amsterdam, 1 edition, 1998.

\bibitem{Ore}
O.~Ore.
\newblock On a special class of polynomials.
\newblock {\em Transactions of the American Mathematical Society}, 35:559--584,
  1993.

\bibitem{Knuth}
{R. L. Graham, and D. E. Knuth, and O. Patashnik}.
\newblock {\em {Concrete Mathematics. A Foundation for Computer Science}}.
\newblock Addison-Wesley Publishing Company, second edition, 1994.

\bibitem{Rosenthal}
A.-L. Trautmann, F.~Manganiello, M.~Braun, and J.~Rosenthal.
\newblock Cyclic orbit codes.
\newblock {\em IEEE Transactions on Information Theory}, 59:7386--7404, 2013.

\end{thebibliography}
	
\end{document}